\documentclass[preprint,12pt]{elsarticle}

\makeatletter
\def\ps@pprintTitle{%
 \let\@oddhead\@empty
 \let\@evenhead\@empty
 \def\@oddfoot{\centerline{\thepage}}%
 \let\@evenfoot\@oddfoot}
\makeatother

\usepackage[margin=.8 in]{geometry}
\usepackage{xcolor}

\usepackage{graphics} % for pdf, bitmapped graphics files
\usepackage{epsfig} % for postscript graphics files
\usepackage{amsmath} % assumes amsmath package installed
\usepackage{amsthm}
\usepackage{amssymb}  % assumes amsmath package installed
\usepackage{subfigure}
\newtheorem{theorem}{Theorem}%[section]
\newtheorem{remark}{Remark}%[section]
\newtheorem{definition}{Definition}
\newtheorem{assumption}{Assumption}

\newcommand{\HEI}[1]{\bf Hybrid-EI}

\begin{document}

\begin{frontmatter}

\title{\textbf{Hybrid Event-Triggered and Impulsive Control for Time-Delay Systems}\tnoteref{label0}}
\tnotetext[label0]{The work of Kexue Zhang was supported by a fellowship from the Pacific Institute for the Mathematical Sciences (PIMS), Canada.}

%\author{Kexue Zhang\corref{cor1}}\ead{kexue.zhang@ucalgary.ca}
\author[cor1]{Kexue Zhang}\ead{kexue.zhang@ucalgary.ca}%\cortext[cor1]{Corresponding author.}
\address[cor1]{Department of Mathematics and Statistics, University of Calgary, Calgary, Alberta T2N 1N4, Canada}

%\author{Bahman Gharesifard\corref{cor2}}\ead{bahman.gharesifard@ucalgary.ca}
\author[cor2]{Bahman Gharesifard}\ead{bahman.gharesifard@ucalgary.ca}
\address[cor2]{Department of Mathematics and Statistics, Queen's University, Kingston, Ontario K7L 3N6, Canada}

\begin{abstract}
In this paper, we study the problem of hybrid event-triggered control for a class of nonlinear time-delay systems. Using a Razumikhin-type input-to-state stability result for time-delay systems, we design an event-triggered control algorithm to stabilize the given time-delay system. In order to exclude Zeno behavior, we combine the impulsive control mechanism with our event-triggered strategy. In this sense, the proposed algorithm is a hybrid impulsive and event-triggered strategy. Sufficient conditions for the stabilization of the nonlinear systems with time delay are obtained by using Lyapunov method and Razumikhin technique. Numerical simulations are provided to show the effectiveness of our theoretical results.
\end{abstract}


\end{frontmatter}

% \linenumbers

\section{Introduction}\label{Sec1}
%\IEEEPARstart{I}{mpulsive} systems, delay-dependent impulses, ISS and research on impulsive systems, what we do in this paper and contribution, organization. 
%I wish you the best of success.
%\subsection{Subsection Heading Here}
%Subsection text here.
%
%% needed in second column of first page if using \IEEEpubid
%%\IEEEpubidadjcol
%
%\subsubsection{Subsubsection Heading Here}
%Subsubsection text here.

{E}{vent-triggered} control strategies allow for updating the control inputs when an event, triggered by a certain event-triggering rule, occurs. Different from the conventional sampled-data control, the unpredictable sequence of event times are determined explicitly by the event-triggering rule. The event-triggering mechanism has the advantage of reducing the number of control input updates while still guaranteeing the underlying desired performance. Therefore, event-triggered control has been widely applied to various control problems, such as consensus problems, distributed optimization protocols, fault detection, and sensor scheduling (see, e.g., the survey papers~\cite{ZPJ-TFL:2015,WPMHH-KHJ-PT:2012} and references therein). 

Time delays are ubiquitous in many practical systems, and dynamical systems with time delay present in many fields (see, e.g. \cite{EF:2014} and references therein). Due to the advantages of event-triggered control in efficiency improvements and the significance of time-delay systems in modeling of the real-world phenomena, the study of event-triggered control strategies for time-delay systems is of great importance. The past few years have witnessed an increased interests in this area; a particular case is the delayed consensus problems of multi-agent systems which are normally described as linear time-delay systems. In~\cite{WZ-ZPJ:2015}, leader-following consensus of multi-agent systems with time delay is studied by employing event-triggered consensus protocols, and Zeno behavior (a phenomenon of infinite number of control updates over a finite time period) was successfully excluded from the this consensus problem. The idea of ruling out Zeno behavior introduced in~\cite{WZ-ZPJ:2015} has been applied to various consensus problems with time-delay, such as observer-based consensus~\cite{DZ-TD:2017} and consensus in stochastic settings~\cite{XT-JC-XL-AA:2017}. The event-triggered consensus protocols considered in~\cite{LL-DWCH-JL:2017,NM-XL-TH:2015} both require the explicit information of the time-delay since the updates of the control signal depend on the delayed state of each agent. Recently, the predictor-based control method has been combined with the event-triggered control mechanism for stabilization of control systems, in order to compensate the delay in the control inputs (see, e.g., \cite{AS-EF:2016-1,AS-EF:2016-2}). Event-triggered stabilization of nonlinear systems in the presence of sensing and actuation delays was studied in \cite{EN-PT-JC:2020}. The provided event-triggering rule requires the memory of system states at some historical moments, and some assumptions on the time-varying delays are essential for the design of the predictor-based event-triggered controllers. It should be pointed out that all the above mentioned results consider time delays in the control inputs, but the continuous dynamics of the uncontrolled systems are free of delays.

A main challenge in the design of event-triggered controllers for time-delay systems is to exclude Zeno behavior.  The well-known result introduced in~\cite{PT:2007} and references thereafter provide an efficient procedure to overcome Zeno behavior in event-triggered control systems, in the absence of delay. However, as expected, and also demonstrated in this paper, with a scalar linear time-delay system, the result cannot be directly applied to time-delay control systems. The literature on event-triggered control of nonlinear systems with time delay is surprisingly limited. We list here the few results in this direction. An event-triggered control scheme has been successfully designed for event-triggering stabilization of a class of nonlinear time-delay systems with discrete and distributed delays in \cite{SD-NM-JFGC:2014}. The proposed event-triggering rule and the feedback controller both require the full knowledge of the system delays, and such requirement plays an important role in ruling out Zeno behavior. Periodic event-triggered control of nonlinear time-delay systems was studied in~\cite{AB-PP-IDL-MDF:2021}. The proposed event-triggering algorithm allows periodic sampling of the systems states, and this guarantees the non-existence of Zeno behavior. Nevertheless, the conditions on the sampling period narrow down the applicable class of time-delay systems, and periodic event-triggered control potentially increases the number of control updates when compared with aperiodic one. The very recent work~\cite{KZ-BG-EB:2021} investigated the event-triggered control problem of nonlinear time-delay systems. { Most crucially related to our current work, the proposed strategy in~\cite{KZ-BG-EB:2021} rules out Zeno behavior, however, while boundedness and attractivity of the closed-loop systems under the proposed event-triggered control algorithm can be guaranteed, \emph{stability of the control systems may not be maintained}.}  

In this paper, we study the event-triggering stabilization problem of time-delay systems from hybrid control point of view. Coupled with the event-triggered control mechanism, we use impulsive controls to help rule out Zeno behavior. The mechanism of impulsive control is to use \textit{impulses}, which are state jumps or abrupt changes over a negligible time period, to achieve the desired performance of the closed-loop system. It has been proved to be powerful in the design and synthesis of control systems and has found applications in many areas, such as, secure communications and treatment of infectious disease (see, e.g., \cite{TY:2001,BMM-EYR:2012,IS-GTS:2016,XZ-PS:2017,ICM-SM-AG-AMG:2016} and the references therein). Our hybrid control algorithm works as follows: we first prescribe a minimal dwell time as the lower bound of each inter-execution time (the time between two successive control updates). If a certain measurement error becomes large enough at a time outside this dwell-time period, then we update the control input. Otherwise, the control signal will not be updated until the end of this dwell-time period, at which time we execute an impulsive control input and then update the feedback control signal. This newly proposed hybrid algorithm ensures the exclusion of Zeno behavior. By using Lyapunov function method and Razumikhin technique, we construct sufficient conditions on the impulse inputs and impulse moments to \emph{preserve the asymptotic stability of the corresponding closed-loop system} under this hybrid control algorithm. 

 The main contribution of this study is threefold. First, a novel hybrid event-triggered and impulsive control algorithm is proposed to asymptotically stabilize a type of nonlinear time-delay systems. Sufficient conditions on the tunable lower bound of inter-execution times are obtained which naturally rules out Zeno behavior. Secondly, without the impulsive control inputs in the proposed hybrid control scheme, it is shown that a particular linear  time-delay system under the event-triggered control algorithm exhibits Zeno behavior which demonstrates the necessity of the impulsive control strategy in our hybrid control algorithm. Lastly, we show that without the event-triggered control inputs, the impulsive control strategy is a special case of our hybrid control scheme, and our hybrid algorithm has the advantage that fewer control updates are triggered than the impulsive control method.

%\margin{I removed the organization. I think it is ok without ( a bit repetitive given the previous part) but we can place it back.}
%The rest of this paper is organized as follows. Sections~II contains some mathematical preliminaries. In Section~III, we introduce a hybrid event-triggered and impulsive control algorithm and conduct the stability analysis of the closed-loop system. In Section~IV, we investigate a linear hybrid control system to demonstrate our theoretical results. Our idea of future research is summarized in Section~V.

\section{Preliminaries}\label{Sec2}

Let $\mathbb{N}$ denote the set of positive integers, $\mathbb{R}$ the set of real numbers, $\mathbb{R}^+$ the set of nonnegative reals, and $\mathbb{R}^n$ the $n$-dimensional real space equipped with the Euclidean norm denoted by $\|\cdot\|$. For $a,b\in \mathbb{R}$ with $b>a$, let $\mathcal{PC}([a,b],\mathbb{R}^n)$ denote the set of piecewise right continuous functions $\varphi:[a,b]\rightarrow\mathbb{R}^n$, and $\mathcal{PC}([a,\infty),\mathbb{R}^n)$ the set of functions $\phi:[a,\infty)\rightarrow\mathbb{R}^n$ satisfying $\phi|_{[a,b]}\in \mathcal{PC}([a,b],\mathbb{R}^n)$ for all $b>a$, where $\phi|_{[a,b]}$ is a restriction of $\phi$ on interval $[a,b]$. For $x\in\mathcal{PC}([a,\infty),\mathbb{R}^n)$, define $\Delta x$ as $\Delta x(t):=x(t^+)-x(t^-)$, where $x(t^+)$ and $x(t^-)$ denote respectively the right- and left-hand limits of $x$ at $t$. Given $\tau>0$, the linear space $\mathcal{PC}([-\tau,0],\mathbb{R}^n)$ is equipped with a norm defined by $\|\varphi\|_{\tau}:=\sup_{s\in[-\tau,0]}\|\varphi(s)\|$ for $\varphi\in \mathcal{PC}([-\tau,0],\mathbb{R}^n)$. For the sake of simplicity, we use $\mathcal{PC}_{\tau}$ to represent $\mathcal{PC}([-\tau,0],\mathbb{R}^n)$ for the rest of this paper.
%\kexue{To indicate the dependency on $ \tau $, we use $\mathcal{PC}_{\tau}$ instead.} \margin{That is good. Could we change the next one right after to "For $\mathcal{PC}_{\tau}$"? Just to be consistent and not to go back to $ a $ again.} 

Consider the following control system with time delay:
\begin{eqnarray}\label{sys0}
\left\{\begin{array}{ll}
\dot{x}(t)=f(t,x_t)+Bu(t), \cr
x_{t_0}=\varphi,
\end{array}\right.
\end{eqnarray}
where $x(t)\in\mathbb{R}^n$ is the system state; $x_{t}$ is defined as $x_{t}(s)=x(t+s)$ for $s\in[-\tau,0]$, and $\tau>0$ is the maximum involved delay; $B\in\mathbb{R}^{n\times m}$ is the control gain; both $f:\mathbb{R}^+\times\mathcal{PC}_{\tau} \rightarrow\mathbb{R}^n$ and the control input $u:[t_0,\infty)\rightarrow\mathbb{R}^m$ satisfy $f(t,0)=u(0)=0$ so that system~\eqref{sys0} admits the trivial solution. In this study, we consider the hybrid impulsive control $u(t)=u_1(t)+u_2(t)$ with the state feedback control 
\begin{equation}\label{feedback}
u_1(t)=k(x(t)),
\end{equation}
where $k:\mathbb{R}^n\rightarrow \mathbb{R}^m$ is the feedback control law, and the impulsive control
\begin{equation}\label{impulsive}
u_2(t)=\sum^{\infty}_{i=1} g(x(t)) \delta(t-s_i),
\end{equation}
where $g:\mathbb{R}^n\rightarrow\mathbb{R}^m$ is the impulsive control law, $\delta(\cdot)$ denotes Delta dirac function, and the sequence of impulse times $\{s_i\}_{i\in\mathbb{N}}$ satisfies $0\leq t_0< s_i $, with $ s_i<s_j $ for $ i<j $, and $\lim_{i\rightarrow\infty}s_i=\infty$. Hence, the closed-loop system~\eqref{sys0} with~\eqref{feedback} and~\eqref{impulsive} can be rewritten as an impulsive system 
\begin{eqnarray}\label{sys}
\left\{\begin{array}{ll}
\dot{x}(t)=f(t,x_t)+Bk(x), ~~t\not=s_i,\cr
\Delta x(s_i)=B g(x(s^-_i)), ~~i\in\mathbb{N},\cr
x_{t_0}=\varphi.
\end{array}\right.
\end{eqnarray}
We refer the reader to~\cite[Chapter 4.1]{XL-KZ:2019} for a detailed discussion on transformation of control system~\eqref{sys0} into impulsive system~\eqref{sys}. Without loss of generality, we assume that $x$ is right-continuous at each impulse time. Define $\bar{f}(t,\phi)=f(t,\phi)+Bk(\phi(0))$, and then assume all the necessary conditions on $\bar{f}$ and $g$ in \cite{GB-XL:1999} hold so that, for any initial condition $\varphi\in\mathcal{PC}_{\tau}$, system~\eqref{sys} has a unique solution $x(t,t_0,\varphi)$ that exists in a maximal interval $[t_0-\tau,t_0+\Gamma)$, where $0<\Gamma\leq \infty$.

The notion of input-to-state stability (ISS), introduced by Sontag in \cite{EDS:1989}, has been proved powerful in the analysis and controller design of dynamical systems, especially in the design of event-triggered controllers (see, e.g. \cite{WZ-ZPJ:2015,PT:2007}). We introduce the following function classes before giving the formal ISS definition. A continuous function $\alpha:\mathbb{R}^+\rightarrow\mathbb{R}$ is said to be of class $\mathcal{K}$ and we write $\alpha\in\mathcal{K}$, if $\alpha$ is strictly increasing and $\alpha(0)=0$. If $\alpha$ is also unbounded, we say that $\alpha$ is of class $\mathcal{K}_{\infty}$ and we write $\alpha\in\mathcal{K}_{\infty}$. A continuous function $\beta:\mathbb{R}^+\times\mathbb{R}^+ \rightarrow\mathbb{R}^+$ is said to be of class $\mathcal{KL}$ and we write $\beta\in \mathcal{KL}$, if $\beta(\cdot,t)\in\mathcal{K}$ for each $t\in\mathbb{R}^+$ and $\beta(s,t)$ decreases to $0$ as $t\rightarrow \infty$ for each $s\in \mathbb{R}^+$. We are now in the position to state the ISS definition for system~\eqref{sys0}.

\begin{definition}
System~\eqref{sys0} is said to be input-to-state {stable} (ISS) with respect to input $u$, if there exist functions $\beta\in\mathcal{KL}$ and $\gamma\in\mathcal{K}_{\infty}$ such that, for each initial condition $\varphi\in\mathcal{PC}_{\tau}$ and input function $u\in\mathcal{PC}([t_0,\infty),\mathbb{R}^m)$, the corresponding solution to~\eqref{sys0} exists globally and satisfies
$$\|x(t)\|\leq \beta(\|\varphi\|_{\tau},t-t_0)+\gamma\Big(\sup_{s\in[t_0,t]}\|u(s)\|\Big), \mathrm{~for~all~} t\geq t_0.$$
\end{definition}
Next, we present several concepts regarding to Lyapunov functions and review a Razumikhin-type ISS result that will be used for the design of our event-triggered control algorithm. A function $V:\mathbb{R}^+\times\mathbb{R}^n\rightarrow \mathbb{R}^+$ is said to be of class $\nu_0$ and we write $V\in\nu_0$, if, for each $x\in\mathcal{PC}(\mathbb{R}^+,\mathbb{R}^n)$, the composite function $t\mapsto V(t,x(t))$ is in $\mathcal{PC}(\mathbb{R}^+,\mathbb{R}^+)$ and can be discontinuous at some $t'\in\mathbb{R}^+$ only when $t'$ is a discontinuity point of $x$. Given a function $V\in \nu_0$ and an input $u\in \mathcal{PC}([t_0,\infty),\mathbb{R}^m)$, the upper right-hand derivative $\mathrm{D}^+V$ of the Lyapunov function candidate $V$ with respect to system~\eqref{sys0} is defined as follows: 
\begin{equation}
\mathrm{D}^+V(t,\phi(0))=\limsup_{h\rightarrow 0^+} \frac{V(t+h,\phi(0)+h ({f}(t,\phi)+Bu) )-V(t,\phi(0))}{h} \nonumber
\end{equation}
for $\phi\in\mathcal{PC}_{\tau}$. We focus here on a special case in which no impulses are considered when system~\eqref{sys0} is free of the control; 
a more general form of system~\eqref{sys0} with external inputs and impulses is studied in~\cite{WHC-WXZ:2009}. The next ISS result, which we recall from~\cite{ART:1998},
is key in our work.

\begin{theorem}\label{Th.ISS}
Assume that there exist functions $V\in \nu_0$ and $\alpha_1, \alpha_2, \chi\in \mathcal{K}_{\infty}$, and constants $q>1$, $c>0$ such that, for all $t\in \mathbb{R}^+$, $x\in \mathbb{R}^n$ and $\phi\in \mathcal{PC}_{\tau}$, 
\begin{itemize}
\item[(i)] $\alpha_1(\|x\|)\leq V(t,x)\leq \alpha_2(\|x\|)$;

\item[(ii)] whenever $q V(t,\phi(0))\geq V(t+s,\phi(s))$ for all $s\in [-\tau,0]$,
\[
\mathrm{D}^+V(t,\phi(0)) \leq -c V(t,\phi(0)) +\chi(\|u\|).
\]
\end{itemize}
Then system~\eqref{sys0} is ISS.
\end{theorem}
It should be noted that Theorem~\ref{Th.ISS} is also true for the general nonlinear time-delay systems, though we stated a version of this theorem tuned to~\eqref{sys0}. We can conclude from Theorem~\ref{Th.ISS} that the global asymptotic stability (GAS) of system~\eqref{sys0} is guaranteed when $u=0$. The Razumikhin-type condition~(ii) in Theorem~\ref{Th.ISS} plays an essential role in the event-triggered controller design, as we demonstrate in the next section. 
\section{Hybrid Event-Triggered Control Algorithm}\label{ETC}

This section introduces an event-triggered feedback control algorithm with impulsive control strategy to asymptotically stabilize time-delay systems.  We propose a two-stage design of the hybrid event-triggered controller. First, we design the event-triggered implementation of control $u_1$ to stabilize the time-delay system in Subsection~\ref{ETC-A}. To ensure the exclusion of Zeno behavior with the proposed event-triggering scheme while still preserving the stability guarantee, a hybrid control algorithm with impulsive controls is provided to determine the impulse times in Subsection~\ref{ETC-B}, and then a stability result with sufficient conditions on the impulsive control input and impulse times is established in Subsection~\ref{ETC-C}.

\subsection{Event-Triggered Feedback Control}\label{ETC-A}

To design the event-triggered implementation of $u_1$, we consider system~\eqref{sys} with $g\equiv 0$ and state-feedback control as follows:
\begin{eqnarray}\label{ETC.sys}
\left\{\begin{array}{ll}
\dot{x}(t)=f(t,x_t)+Bu_1(t), \cr
u_1(t)=k(x(t_i)),~t\in[t_i,t_{i+1})\cr
x_{t_0}=\varphi,
\end{array}\right.
\end{eqnarray}
where $u_1\in \mathbb{R}^m$ is the feedback control input, and $k: \mathbb{R}^n\rightarrow \mathbb{R}^m$ is the feedback control law. The time sequence $\{t_i\}_{i\in\mathbb{N}}$ is implicitly defined by a certain execution rule to be determined later based on the measurement of system states, and each time instant $t_i$ corresponds to a control update $u_1(t_i)$. To be more specific, the controller $u_1$ samples the system states and updates its input signal both at each $t_i$ while remaining constant between two successive control updates.

Let us define the state measurement error by
\begin{equation}\label{error}
\mathbf{e}(t)=x(t_i)-x(t),
\end{equation}
for $t\in [t_i,t_{i+1})$ with $i\in\mathbb{N}$, and then rewrite
\begin{equation}\label{ETCer}
u_1(t)=k(x(t_i))=k(\mathbf{e}(t)+x(t)).
\end{equation}
Substituting~\eqref{ETCer} into system~\eqref{ETC.sys} gives the following closed-loop system:
\begin{eqnarray}\label{CL.sys}
\left\{\begin{array}{ll}
\dot{x}(t)=f(t,x_t)+B k(\mathbf{e}+x), \cr
x_{t_0}=\varphi.
\end{array}\right.
\end{eqnarray}
We make the following assumption on the control system~\eqref{CL.sys}.
\begin{assumption}\label{A1}
There exist functions $V\in \nu_0$ and $\alpha_1, \alpha_2, \chi\in \mathcal{K}_{\infty}$, and constants $q>1$, $c>0$ such that all the conditions of Theorem~\ref{Th.ISS} hold for system~\eqref{CL.sys} with input $u$ replaced with $\mathbf{e}$.
\end{assumption}
It can be seen from Theorem~\ref{Th.ISS} that Assumption~\ref{A1} guarantees that closed-loop system~\eqref{CL.sys} is ISS with respect to measurement error $\mathbf{e}$, and system~\eqref{CL.sys} is GAS provided $\mathbf{e}=0$. In this paper, we design an execution rule to determine the time sequence $\{t_i\}_{i\in\mathbb{N}}$ for the updates of the feedback control $u_1$ so that closed-loop system~\eqref{CL.sys} with measurement error $\mathbf{e}$ is still GAS. To do so, we restrict $\mathbf{e}$ to satisfy
\begin{equation}\label{exrule}
\chi(\|\mathbf{e}\|)\leq \sigma \alpha_1(\|x\|) \textrm{~for~ some~} \sigma>0,
\end{equation}
then the dynamics of $V$ is bounded by 
\begin{align}\label{DV}
\mathrm{D}^+V(t,x) \leq -cV(t,x) +\sigma \alpha_1(\|x\|)\leq -(c-\sigma)V(t,x)
\end{align} 
whenever $qV(t,x(t))\geq V(t+s,x(t+s))$ for all $s\in [-\tau,0]$. This guarantees the control system~\eqref{CL.sys} is GAS provided $\sigma<c$. The updating of the control input $u_1$ can be triggered by the execution rule (or event)
\begin{equation}\label{event}
\chi(\|\mathbf{e}\|) = \sigma \alpha_1(\|x\|).
\end{equation}
The event times are the instants when the event happens, that is, 
\begin{equation}\label{et.time}
t_{i+1}=\inf\{t\geq t_i \mid \chi(\|\mathbf{e}\|) = \sigma \alpha_1(\|x\|)\}.
\end{equation}
According to the control law in~\eqref{ETC.sys}, the control input is updated at each $t_i$ (the error $\mathbf{e}$ is set to zero simultaneously) and remains constant until the next event time $t_{i+1}$, and then the error $\mathbf{e}$ is reset to zero again. Therefore, the proposed event times in~\eqref{et.time} ensures the GAS of control system~\eqref{ETC.sys}.

Since the event times in~\eqref{et.time} are defined implicitly, it is essential to rule out the existence of Zeno behavior, which we define below for completeness. 
\begin{definition}[Zeno Behavior]
If there exists $T>0$ such that $t_l\leq T$ for all $l\in \mathbb{N}$, then system~\eqref{ETC.sys} is said to exhibit Zeno behavior.
\end{definition}

It worth mentioning that for control systems without time-delay, a well-known result in \cite{PT:2007} says that if the functions $f$ and $k$ in~\eqref{ETC.sys} are Lipschitz continuous on compact sets, then it is possible to exclude Zeno behavior (some extra conditions are required for a definite exclusion, see \cite{PT:2007} for more details). However, this is not true for event-triggered control systems with time-delay. { Even though this is to be expected,} an example is discussed with numerical simulations in Section~\ref{Sec4} to illustrate this statement.

\subsection{Excluding Zeno Behavior via Impulsive Control}\label{ETC-B}

As we demonstrated in Section~\ref{Sec4}, a linear control system with time-delay can exhibit Zeno behavior under natural event-triggered control strategies. Our next objective is to show that one can still use event-triggered strategies for system~\eqref{sys}, as long as they are paired with impulsive control $u_2$, designed to exclude Zeno behaviors. 

To proceed, let us define a sequence of event-time candidates
\begin{equation}\label{ET-candidate}
\bar{t}_{i+1}=\inf\{t\geq t_i \mid \chi(\|\mathbf{e}\|)=\sigma \alpha_1(\|x\|)\},
\end{equation}
where the sequence of event times $\{t_i\}_{i\in\mathbb{N}}$ is to be determined with a lower bound $h>0$ of the inter-execution time $\inf_{i\in \mathbb{N}}\{t_{i+1}-t_i\}$ according to the execution rule displayed at \textbf{Hybrid-EI} below. 
\begin{table}[h]
\hskip-3mm\rule{\textwidth}{1.pt}
\vskip-0mm
\hskip-3mm\textbf{Hybrid Event-triggered/Impulsive Strategy (Hybrid-EI)}
\vskip-1.5mm
\hskip-3mm\rule{\textwidth}{1.pt}
\vskip-1mm
\begin{itemize}
%\vskip-1mm
\item[1.] If $\bar{t}_{i+1}-t_i>h$, then let $t_{i+1}=\bar{t}_{i+1}$ and update the feedback control signal $u_1$ at $t=t_{i+1}$.
\bigskip
\item[2.] If $\bar{t}_{i+1}-t_i\leq h$, then activate an impulse input $\Delta x(t)=B g(x(t))$ at $t=t_i+h$. Let $t_{i+1}=t_i+h$ and update the control $u_1$ at $t=t_{i+1}$ after the state jump.
\end{itemize}
\vskip-1mm
\hskip-3mm\rule{\textwidth}{1.pt}
\end{table}
%\begin{itemize}
%\item if $\bar{t}_{i+1}-t_i>h$, then let $t_{i+1}=\bar{t}_{i+1}$ and update the control input signal at $t=t_{i+1}$;
%
%\item if $\bar{t}_{i+1}-t_i\leq h$, then activate an impulse (state jump) $\Delta x(t)=I(t,x)$ at $t=t_i+h$ where $I:\mathbb{R}^+\times \mathbb{R}^n\rightarrow \mathbb{R}^n$ regulates the state jump. Let $t_{i+1}=t_i+h$ and update the control input at $t=t_{i+1}$.
%\end{itemize}

It can be seen from \textbf{\HEI} that the inter-execution times $\{t_{i+1}-t_i\}_{i\in\mathbb{N}}$ are lower bounded by $h$, that is, $t_{i+1}-t_i\geq h$ for all $i\in \mathbb{N}$. This excludes Zeno behavior. The closed-loop system can be written as the following impulsive system:
\begin{eqnarray}\label{isys}
\left\{\begin{array}{ll}
\dot{x}(t)=f(t,x_t)+Bu_1(t), \cr
u_1(t)=k(x(t_i)),~t\in[t_i,t_{i+1})\cr
\Delta x(t_{i+1})=B g(x(t^-_{i+1})),~\textrm{~if~} t_{i+1}=t_i+h \cr 
x_{t_0}=\varphi.
\end{array}\right.
\end{eqnarray}
If $t_{i+1}=t_i+h$, then $t_{i+1}$ is an impulse time. We assume that system~\eqref{isys} satisfies all the necessary conditions in \cite{GB-XL:1999} so that for any initial condition $\varphi\in\mathcal{PC}_{\tau}$, system $\eqref{isys}$ has a unique global solution $x(t,t_0,\varphi)$.

\subsection{Stability Analysis}\label{ETC-C}

To ensure the asymptotic stability of system~\eqref{isys} with the proposed hybrid event-triggered control algorithm, the next theorem presents several sufficient conditions on the continuous dynamics of system~\eqref{isys}, the impulses, the lower bound of inter-execution times, and the relation among them.

\begin{theorem}\label{Th.main}
Suppose that assumption~\ref{A1} holds with $V\in \nu_0$, $q>1$ and $c>0$. For some $h>0$, the event times $\{t_i\}_{i\in\mathbb{N}}$ are defined according to \textbf{\HEI} with the event-time candidates given in~\eqref{ET-candidate} and positive constant $\sigma<c$. If $t_{i+1}=t_i+h$, we further assume that there exist positive constants $\bar{c}$ and $\rho$ such that the following conditions are satisfied
\begin{itemize}
\item[(i)] for $t\in [t_i,t_{i+1})$, $
\mathrm{D}^+V(t,x) \leq \bar{c} V(t,x)
$,
whenever $qV(t,x(t))\geq V(t+s,x(t+s))$, for all $s\in [-\tau,0]$;

\item[(ii)] $V(t_{i+1},x(t^-_{i+1})+B g(x(t^-_{i+1})))\leq \rho V(t^-_{i+1},x(t^-_{i+1}));$

\item[(iii)] $q>\frac{1}{\rho}>e^{\bar{c}h}$.

\end{itemize}
Then the closed-loop system~\eqref{isys} is GAS.

\end{theorem}

\begin{proof}
Condition (iii) implies that there exists a small enough $\varepsilon>0$ such that $\frac{1}{\rho}>\frac{1}{\rho+\varepsilon}>e^{\bar{c}h}$. We then can find a positive constant $\lambda$ close to zero and $\lambda\leq c-\sigma$ so that both
\[
q\geq \frac{e^{\lambda\tau}}{\rho}>\frac{1}{\rho}>e^{(\bar{c}+\lambda)h}
\textrm{~and~}
q > \frac{e^{\lambda\tau}}{\rho+\varepsilon}>\frac{1}{\rho+\varepsilon}>e^{(\bar{c}+\lambda)h}
\]
are satisfied. Let $M= qe^{-\lambda\tau}$, then we have $M>1$ and $(\rho+\varepsilon)M>1$. Let $v(t)=V(t,x(t))$ and define $w(t)=e^{\lambda(t-t_0)}v(t)$ for $t\geq t_0-\tau$. By induction, we will show that, for $t\in[t_i,t_{i+1})$,
\begin{eqnarray}\label{induction}
w(t)\leq
\left\{\begin{array}{ll}
(\rho+\varepsilon)M\alpha_2(\|\varphi\|_{\tau}), \textrm{~if~} t_{i+1}>t_i+h\cr
M\alpha_2(\|\varphi\|_{\tau}), \textrm{~if~} t_{i+1}=t_i+h
\end{array}\right.
\end{eqnarray}
For $s\in [-\tau,0]$, we have 
\begin{align*}
w(t_0+s)=e^{\lambda(t_0+s-\tau)}v(t_0+s)\leq v(t_0+s)\leq \alpha_2(\|\varphi\|_{\tau}) <(\rho+\varepsilon)M\alpha_2(\|\varphi\|_{\tau}).
\end{align*} 
Therefore,~\eqref{induction} is true on $[t_0-\tau,t_0]$. We now prove~\eqref{induction} holds on $[t_0,t_1)$. To do this, we consider the following two cases.

\underline{Case I:} $t_1>t_0+h$. We will show that 
\begin{equation}\label{indcution1}
w(t)\leq (\rho+\varepsilon)M\alpha_2(\|\varphi\|_{\tau})
\end{equation}
is true for $t\in[t_0,t_1)$. We do this by contradiction. Suppose~\eqref{indcution1} is not true on $[t_0,t_1)$, then there exists some $t\in (t_0,t_1)$ so that $w(t)> (\rho+\varepsilon)M\alpha_2(\|\varphi\|_{\tau})$.
Define now 
\[
t^*=\inf\{t\in (t_0,t_1) \mid w(t)> (\rho+\varepsilon)M\alpha_2(\|\varphi\|_{\tau}) \}.
\] 
By the continuity of $w$, we conclude that 
\[
w(t^*)= (\rho+\varepsilon)M\alpha_2(\|\varphi\|_{\tau})
\textrm{~and~} 
w(t)\leq (\rho+\varepsilon)M\alpha_2(\|\varphi\|_{\tau})
\]
for all $t\in [t_0,t^*]$. We next define 
\[
t^{**}=\sup\{t\in[t_0,t^*]\mid w(t)\leq \alpha_2(\|\varphi\|_{\tau})\}.
\] 
Since $w(t_0)\leq \alpha_2(\|\varphi\|_{\tau})$ and $w(t^*)> \alpha_2(\|\varphi\|_{\tau})$, we conclude that 
\[
w(t^{**})=\alpha_2(\|\varphi\|_{\tau}) \textrm{~and~} \alpha_2(\|\varphi\|_{\tau})\leq w(t)\leq (\rho+\varepsilon)M\alpha_2(\|\varphi\|_{\tau})
\] 
for $t\in[t^{**},t^{*}]$. Thus, for any $t\in[t^{**},t^{*}]$, we have $t+s\leq t^*$ for all $s\in [-\tau,0]$ and 
\[
w(t+s)\leq (\rho+\varepsilon)M\alpha_2(\|\varphi\|_{\tau}) \leq (\rho+\varepsilon)M w(t)
\]
which implies
\begin{align}\label{precon1}
v(t+s) < (\rho+\varepsilon)M w(t) e^{-\lambda(t+s-t_0)} < (\rho+\varepsilon)M e^{\lambda\tau} v(t) < q v(t).
\end{align}
Here, we used the facts that $(\rho+\varepsilon)M>1$, $\rho+\varepsilon<1$ and $M=q e^{-\lambda\tau}$. We then can conclude from condition (ii) of Theorem~\ref{Th.ISS} and~\eqref{DV} that for $t\in[t^{**},t^{*}]$
\begin{align*}
\mathrm{D}^+ w(t) = \lambda e^{\lambda(t-t_0)} v(t) +e^{\lambda(t-t_0)}\mathrm{D}^+ v(t) = (\lambda-c+\sigma) w(t)\leq 0.
\end{align*}
This indicates that $w(t)$ is nonincreasing on $[t^{**},t^{*}]$ and then $w(t^{**})\geq w(t^*)$ which is a contradiction to the definitions of $t^*$ and $t^{**}$. Hence,~\eqref{indcution1} is true on $[t_0,t_1)$ for Case~I.

\underline{Case II:} $t_1=t_0+h$. We will show that 
\begin{equation}\label{indcution2}
w(t)\leq M\alpha_2(\|\varphi\|_{\tau})
\end{equation}
holds on $[t_0,t_1)$ by contradiction. Assume that~\eqref{indcution2} does not hold, then there exists some $t\in[t_0,t_1)$ such that $w(t)> M\alpha_2(\|\varphi\|_{\tau})$, and we define 
\[
\bar{t}=\inf\{t\in[t_0,t_1) \mid w(t)> M\alpha_2(\|\varphi\|_{\tau})\}. 
\]
It follows from the continuity of $w$ on $[t_0,t_1)$ that $w(\bar{t})=M\alpha_2(\|\varphi\|_{\tau})$ and $w(t)\leq M\alpha_2(\|\varphi\|_{\tau})$ for all $t\in [t_0,\bar{t}]$. Since $w(t_0)\leq \alpha_2(\|\varphi\|_{\tau})$, there exists a $t\in [t_0,\bar{t})$ so that $w(t)>\alpha_2(\|\varphi\|_{\tau})$. Let 
\[
\tilde{t}=\sup\{t\in[t_0,\bar{t})\mid w(t)\leq \alpha_2(\|\varphi\|_{\tau})\}.
\]
Then we have $w(\tilde{t})=\alpha_2(\|\varphi\|_{\tau})$ and $\alpha_2(\|\varphi\|_{\tau})\leq w(t)\leq M \alpha_2(\|\varphi\|_{\tau})$, for all $t\in [\tilde{t},\bar{t}]$. Therefore, for $t\in [\tilde{t},\bar{t}]$, we have $w(t+s)\leq M \alpha_2(\|\varphi\|_{\tau})\leq M w(t)$ for all $s\in[-\tau,0]$, which then implies $v(t+s)\leq e^{-\lambda(t+s-t_0)}Mw(t) =  M e^{-\lambda s}v(t)\leq q v(t)$.

Condition (ii) states that
\begin{align}\label{precon2}
\mathrm{D}^+w(t) = \lambda e^{\lambda(t-t_0)} v(t) +e^{\lambda(t-t_0)} \mathrm{D}^+v(t) \leq \lambda w(t) +\bar{c} e^{\lambda(t-t_0)}v(t) = (\lambda+\bar{c}) w(t) 
\end{align}
for all $t\in[\tilde{t},\bar{t}]$, then $w(\bar{t})\leq w(\tilde{t})e^{(\lambda+\bar{c})(\bar{t}-\tilde{t})}\leq w(\tilde{t})e^{(\lambda+\bar{c})h}$, in which we used $\bar{t}-\tilde{t}<h$. From the definition of $\tilde{t}$ and the fact $e^{(\lambda+\bar{c})h}< q e^{-\lambda\tau}=M$, we get $w(\bar{t})< M \alpha_2(\|\varphi\|_{\tau})$. This produces a contradiction to the definition of $\bar{t}$. Therefore,~\eqref{indcution2} is true on $[t_0,t_1)$ for this case. We hence conclude from the above two cases that~\eqref{induction} holds for $t\in [t_0,t_1)$.

Now suppose~\eqref{induction} holds on $[t_0,t_m)$ where $m\geq 1$, and we next prove that~\eqref{induction} is still true for $t\in [t_m,t_{m+1})$. Similar to the above discussion, we consider two scenarios.

\underline{Case I':} $t_{m+1}>t_m+h$. We have $w(t)\leq M \alpha_2(\|\varphi\|_{\tau})$ for all $t< t_m$ from~\eqref{induction}, and will prove~\eqref{indcution1} holds on $[t_m,t_{m+1})$. For $t=t_m$, we have 
\[
w(t_m)\leq \rho w(t^-_m)\leq \rho M \alpha_2(\|\varphi\|_{\tau})< (\rho+\varepsilon)M \alpha_2(\|\varphi\|_{\tau}),
\] 
that is,~\eqref{indcution1} is true for $t=t_m$. We next show~\eqref{indcution1} is true on $(t_m,t_{m+1})$ by a contradiction argument. Suppose~\eqref{indcution1} does not hold, then we can find a $t\in (t_m,t_{m+1})$ so that $w(t)>(\rho+\varepsilon)M \alpha_2(\|\varphi\|_{\tau}) $. To proceed, we define 
\[
t_*=\inf\{t\in (t_m,t_{m+1}) \mid w(t)>(\rho+\varepsilon)M \alpha_2(\|\varphi\|_{\tau})\}.
\]
Using the facts that $w(t_m)<(\rho+\varepsilon)M \alpha_2(\|\varphi\|_{\tau})$ and the continuity of $w$ on $[t_m,t_{m+1})$, we conclude
\[
w(t_*)=(\rho+\varepsilon)M \alpha_2(\|\varphi\|_{\tau}) 
\textrm{~and~} 
w(t)\leq (\rho+\varepsilon)M \alpha_2(\|\varphi\|_{\tau})
\] 
on $[t_m,t_*]$. We further define 
\[
t_{**}=\sup\{t\in [t_m,t_*)\mid w(t)\leq \rho M \alpha_2(\|\varphi\|_{\tau})\}.
\] 
Since $w(t_m)\leq \rho M \alpha_2(\|\varphi\|_{\tau})$ and $w(t_*)>\rho M \alpha_2(\|\varphi\|_{\tau})$, we conclude from the continuity of $w$ that $w(t_{**})=\rho M \alpha_2(\|\varphi\|_{\tau})$
and 
\[
\rho M \alpha_2(\|\varphi\|_{\tau})\leq w(t)\leq (\rho+\varepsilon)M \alpha_2(\|\varphi\|_{\tau})
\] 
on $[t_{**},t_*]$. For $s\in [-\tau,0]$, we have $t+s\leq t_*$ when $t\in [t_{**},t_*]$, then $w(t+s)\leq M \alpha_2(\|\varphi\|_{\tau})$. Therefore, $w(t+s)\leq \frac{1}{\rho}w(t)$ for $t\in [t_{**},t_*]$, which implies 
\begin{align*}
v(t+s) = w(t+s) e^{-\lambda(t+s-t_0)} \leq \frac{1}{\rho} w(t) e^{-\lambda(t-t_0)}e^{-\lambda s} \leq \frac{e^{\lambda \tau}}{\rho} v(t),
\end{align*}
and then we conclude from the fact $\frac{e^{\lambda \tau}}{\rho} \leq q$ that $v(t+s)\leq q v(t)$. Similar to Case~I, we can derive the contradiction: $w(t_*)\leq w(t_{**})$. Thus, we conclude that~\eqref{indcution1} is true on $[t_m,t_{m+1})$ for Case I'.

\underline{Case II':} $t_{m+1}=t_m+h$. We will show~\eqref{indcution2} holds on $[t_m,t_{m+1})$. Based on the assumption that \eqref{induction} holds on $[t_0,t_m)$,
%\kexue{This is the assumption we made for mathematical reduction on $[t_0,t_{m})$. Here we are going to show~\eqref{indcution2} holds on $[t_m,t_{m+1})$ with Case II'.}, 
we have $w(t)\leq M \alpha_2(\|\varphi\|_{\tau})$ for $t\in[t_0,t_m)$. When $t=t_m$, it follows from condition (ii) that 
\[
w(t_m)\leq \rho w(t^-_m)\leq \rho M \alpha_2(\|\varphi\|_{\tau})< M \alpha_2(\|\varphi\|_{\tau})
\] 
which implies~\eqref{indcution2} is true at $t=t_m$. We next will prove~\eqref{indcution2} holds on $(t_m,t_{m+1})$ by contradiction. Suppose there exist a $t\in (t_m,t_{m+1})$ so that $w(t)> M \alpha_2(\|\varphi\|_{\tau})$. We define 
\[
\underline{t}=\inf\{t\in(t_m,t_{m+1})\mid w(t)>M \alpha_2(\|\varphi\|_{\tau})\}. 
\]
The continuity of $w$ yields that $w(\underline{t})=M \alpha_2(\|\varphi\|_{\tau})$ and $w(t)\leq M \alpha_2(\|\varphi\|_{\tau})$, for all $t\leq \underline{t}$. Let 
\[
\hat{t}=\sup\{t\in[t_m,\underline{t})\mid w(t)\leq \rho M \alpha_2(\|\varphi\|_{\tau})\}. 
\] 
Then the facts that $w(t_m)\leq \rho M \alpha_2(\|\varphi\|_{\tau})$ and $w(\underline{t})> \rho M \alpha_2(\|\varphi\|_{\tau})$ imply that $w(\hat{t})=\rho M \alpha_2(\|\varphi\|_{\tau})$ and $\ w(t)\geq \rho M \alpha_2(\|\varphi\|_{\tau})$ on $[\hat{t},\underline{t}]$. Therefore, for $s\in[-\tau,0]$ and $t\in[\hat{t},\underline{t}]$, it follows that 
\[
w(t+s)\leq M \alpha_2(\|\varphi\|_{\tau})\leq \frac{1}{\rho}w(t)\leq M w(t).
\] 
By an argument similar to the one in Case~II, we have that $v(t+s)\leq qv(t)$ for $s\in [-\tau,0]$; therefore,~\eqref{precon2} holds on $[\hat{t},\underline{t}]$ that implies $w(\underline{t})< M \alpha_2(\|\varphi\|_{\tau})$. This is a contradiction with the definition of $\underline{t}$. Hence,~\eqref{indcution2} is true on $(t_m,t_{m+1})$. This completes the induction proof for all $t\geq t_0$. Therefore, $w(t)\leq M\alpha_2(\|\varphi\|_{\tau})$ on $[t_0,\infty)$, as claimed and the global asymptotic stability of closed-loop system~\eqref{isys} hence follows. 
\end{proof}

%\margin{I added one of the remarks as part of the text, and added a part to the next one. Please check.}
Some remarks are in order. Enforcing a lower bound of the inter-execution times in \textbf{\HEI} is similar to the mechanism of time-regularized event-triggered control, where the event triggering condition is only verified after some pre-specified time period has elapsed (see, \cite{HY-FH-TC:2019} and the references therein), leading to conservative criteria on this pre-specified time period. However, our hybrid control algorithm has the advantage that when $t_{i+1}$ is an impulse time, the Razumikhin-type condition~\eqref{DV} may potentially fail to hold on $[t_i,t_{i+1})$, leading to a less conservative criterion, with the caveat that we could possibly have no stability certificate for~\eqref{isys} without a careful tuning of the impulses. Razumikhin-type condition (i) characterizes such continuous dynamics with positive constant $\bar{c}$ indicating the destabilizing effects on $[t_i,t_{i+1})$. Condition (ii) describes the jump that occurs for the Lyapunov function by each stabilizing impulse input, quantified by constant $\rho$. The inequality in condition (iii) balances the continuous dynamics and the impulses so that the closed-loop system~\eqref{isys} is GAS. It can be seen that constant $\rho$ is upper bounded by a quantity depending on both $\bar{c}$ and $h$. Increasing $h$ allows the destabilizing continuous dynamics to evolves over a prolonged time period, and then reduces the upper bound of $\rho$ which corresponds to enlarging the jump of Lyapunov function $V$ at impulse time $t_{i+1}$.

\begin{remark}\label{remark.impulse} 
Compared with the existing results on event-triggered control of time-delay systems mentioned in Section~\ref{Sec1}, the event times determined by \textbf{\HEI} are aperiodic, time delays are considered in the continuous dynamics of the uncontrolled systems. Furthermore, \textbf{\HEI} does not require the knowledge of these system delays other than the boundedness. More specifically, Razumikhin condition (ii) in Theorem~\ref{Th.main} only requires the system delays to be bounded, but the exact bound $\tau$ is not essential for the design of the impulsive control (see Subsection~\ref{subsecIV-B} for a demonstration). Nevertheless, $\tau$ is closely related to the exponential convergence rate of the Lyapunov function (see the selection of $\lambda$ in the above proof).
Finally, without the state-feedback control $u$, the proposed hybrid control scheme reduces to the impulsive control method, that is, $t_{i+1}-t_i=h$ for all $i\in \mathbb{N}$, and $\{t_i\}_{i\in\mathbb{N}}$ is the impulse time sequence. Theorem~\ref{Th.main} provides an impulsive stabilization criterion on impulsive system~\eqref{isys} with $u_1\equiv 0$. Similar Razumikhin-type stability results can be found in \cite{SP-YZ:2010} for stochastic impulsive systems. Compared with the impulsive control strategy, the advantage of our hybrid control algorithm is that the number of impulse times can be reduced significantly due to the intervention of the event-triggered controller, and far fewer control updates are triggered. See Section~\ref{Sec4} for a demonstration with numerical simulations in Fig.~\ref{fig5}.
\end{remark}
\section{An Illustrative Example.}\label{Sec4}

To demonstrate Theorem~\ref{Th.main} with \textbf{\HEI}, we study the following hybrid control system with time-delay
\begin{eqnarray}\label{linear.sys}
\left\{\begin{array}{ll}
\dot{x}(t)=b x(t-r)+u_1(t), ~t\not=s_i \cr
\Delta x(s_i)=\beta x(s^-_i), ~i\in\mathbb{N} \cr
x_{t_0}=\phi,
\end{array}\right.
\end{eqnarray}
where state $x\in \mathbb{R}$, $\phi(s)=1$ for $s\in [-r,0]$, $r=16$, $b=-0.1$, and $u_1(t)=kx(t)$ is the state feedback control with $k=-0.2$. Constant $\beta$ is the impulsive control gain, and impulse time sequence $\{s_i\}_{i\in\mathbb{N}}$ is to be determined according to our hybrid control algorithm. By considering Lyapunov function $V(x)=x^2$, it can be derived from Theorem~\ref{Th.ISS} that control system~\eqref{linear.sys} without impulses (or with $\beta =0$) is GAS. However, system~\eqref{linear.sys} is unstable without all inputs (i.e., $u_1=\beta=0$), and see Fig.~\ref{unstable} for the trajectory of such a system. This instability can be readily verified by the numerical simulation, which is omitted due to the space limitation. The detailed stability and instability analysis of such systems can be found in \cite{HM-SSK:2020} and the references therein.

\begin{figure}[!t]\centering
\includegraphics[width=3.2in]{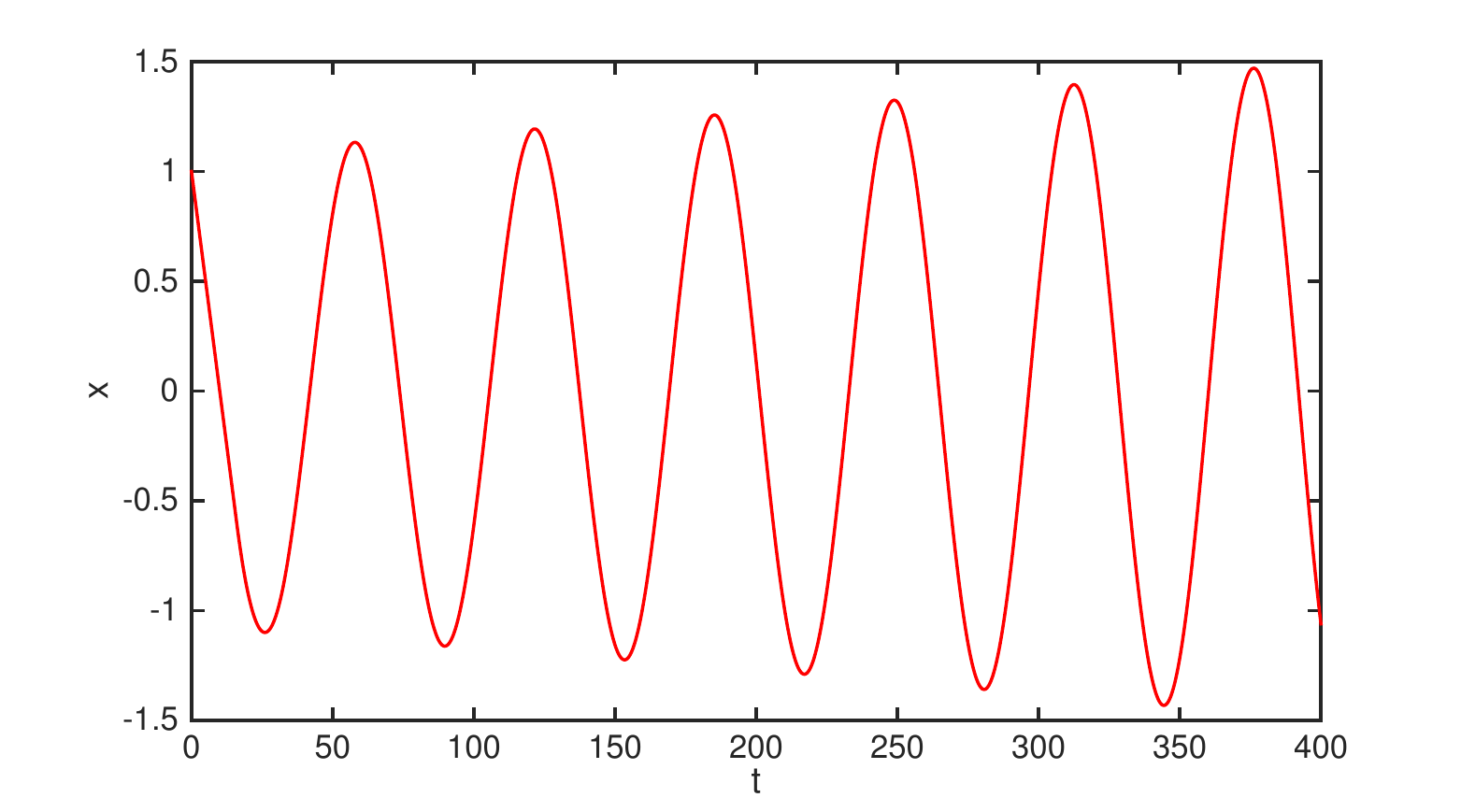}%
\caption{Trajectory of system~\eqref{linear.sys} with $u_1=\beta=0$. }
\label{unstable}
\end{figure}

In the following subsection, we will show the existence of Zeno behavior in control system~\eqref{linear.sys} with event-triggered implementation of $u_1$ and $\beta=0$.

\subsection{Event-triggered Control Method with Zeno Behavior}

Consider the event-triggered implementation of $u_1$ in system~\eqref{linear.sys} without impulses. Then the closed-loop system can be rewritten in the form of~\eqref{CL.sys}:
\begin{eqnarray}\label{linear.CLsys}
\left\{\begin{array}{ll}
\dot{x}(t)=b x(t-r)+kx(t)+k\mathbf{e}(t), \cr
x_{t_0}=\phi,
\end{array}\right.
\end{eqnarray}
where $\mathbf{e}(t)=x(t_i)-x(t)$ for $t\in [t_i,t_{i+1})$ with $i\in \mathbb{N}$, and the sequence of event times $\{t_i\}_{i\in\mathbb{N}}$ is to be determined according to~\eqref{et.time}. Consider the Lyapunov function $V(x)=x^2$.  Then condition (i) of Theorem~\ref{Th.ISS} is satisfied with $\alpha_1(|x|)=\alpha_2(|x|)=x^2$. Using~\eqref{linear.CLsys}, it follows that
%\kexue{To distinguish from Euler's number $e$, we used $\mathbf{e}$ to denote the measurement error throughout this paper.}
\begin{align*}
\dot{V}(x(t))&  =  2kx^2+2bxx(t-r)+2kx\mathbf{e}\cr
          &\leq 2kx^2+|b|(\varepsilon x^2+\varepsilon^{-1}x^2(t-r))+ |k|(\epsilon x^2+\epsilon^{-1}\mathbf{e}^2)\cr
          &  =  (2k+\varepsilon|b|+\epsilon|k|)V(x(t))+\varepsilon^{-1}|b|V(x(t-r))+\epsilon^{-1}|k|\mathbf{e}^2
\end{align*}
where Young's inequality was used twice with $\varepsilon=\sqrt{q}$ and $\epsilon=\sigma/|k|$. Whenever $qV(x(t))\geq V(x(t+s))$ for all $s\in [-r,0]$ with some $q>1$, we have
\begin{equation}
\dot{V}(x)=(2k+\varepsilon|b|+\varepsilon^{-1}q|b|+\epsilon|k|)V(x) + \epsilon^{-1}|k|\mathbf{e}^2,\nonumber
\end{equation}
then, condition (ii) of Theorem~\ref{Th.ISS} holds with 
\[
c=-(2k+\varepsilon|b|+\varepsilon^{-1}q|b|+\epsilon|k|)>0 \ \mathrm{and} \
\chi(|e|)=\epsilon^{-1}|k|\mathbf{e}^2.
\]
Therefore, assumption~\ref{A1} is satisfied for system~\eqref{linear.CLsys}. The event times defined by~\eqref{et.time} are as follows:
\begin{equation}\label{et.time.linear}
t_{i+1}=\inf\{t\geq t_i \mid  \epsilon^{-1}|k|\mathbf{e}^2 = \sigma x^2\},
\end{equation}
where positive constant $\sigma$ satisfies $\sigma<c$. Using the fact $\epsilon=\sigma/|k|$, we can rewrite~\eqref{et.time.linear} as
\begin{equation}\label{et.time.linear1}
t_{i+1}=\inf\{t\geq t_i \mid~  \mathbf{e}^2 = \sigma_0 x^2\}
\end{equation}
and the condition $\sigma<c$ as
\begin{equation}\label{inequ}
2k+\varepsilon|b|+\varepsilon^{-1}q|b|+\epsilon|k|+\epsilon^{-1}|k|\sigma_0<0,
\end{equation}
where $\sigma_0=\epsilon^2$. Then,
\begin{equation}\label{inequ1}
k+\sqrt{q}|b|+\sqrt{\sigma_0}|k|<0
\end{equation}
implies~\eqref{inequ}, by the facts that $\varepsilon=\sqrt{q}$ and $\epsilon=\sigma/|k|$. According to our analysis on system~\eqref{CL.sys}, if~\eqref{inequ1} holds, then closed-loop system~\eqref{linear.CLsys} is GAS with the event times determined by~\eqref{et.time.linear1}.

%\margin{You could decide to remove this plot; I guess it should be easy to see/show this and so we may remove it.}
%\marginb{Is this comment for Fig.1? The instability of system (10) without $u$ is not straightforward, so this figure is provided to show this property.}
%\begin{figure}[t]\centering
%\includegraphics[width=2.5in]{fig1a.eps}%
%\caption{Trajectory of system~\eqref{linear.sys} with $u=0$.}
%\label{fig1}
%\end{figure}
%
%\begin{figure}[t]\centering
%\includegraphics[width=2.5in]{Zenoproof.pdf}%
%\caption{Trajectory of system~\eqref{linear.CLsys} on $[t_0,t_1]$.}
%\label{fig2}
%\end{figure}

The main reason of studying linear scalar system~\eqref{linear.CLsys} is that the existence of Zeno behavior can be verified analytically. To do so, we choose $\sigma_0=0.36$ so that~\eqref{inequ1} is satisfied. We will show there are infinitely many event times over the time interval $[0,10]$. We first prove that $t_1<10$ and $0<x(t_1)<1-t_1/10$. For $t\leq r$, we have $x(t-r)=x(t_0)=1$ and $\dot{x}=-0.1x(t-r)-0.2x(t_0)=-0.3<-0.1$, 
for $t\leq \min\{t_1,r\}$. Therefore, both $x$ and $\mathbf{e}$ are positive, $x$ is strictly decreasing, and $\mathbf{e}$ is strictly increasing for $t\in (t_0,\min\{t_1,r\})$. By~\eqref{et.time.linear1}, $\mathbf{e}^2(t_1)=\sigma_0 x^2(t_1)$, that is, $\mathbf{e}(t_1)=0.6x(t_1)$ which implies $x(t_1)=x(t_0)/1.6$. We then can conclude that $x$ strictly decreases from $1$ at $t_0$ to $x(t_0)/1.6$ at $t_1$ with $t_1<10$, and the decreasing rate is smaller than $-0.1$. Thus, $0<x(t_1)<1-t_1/10$ (see Fig.~\ref{Zenoproof} for an illustration of the above discussion). Next, suppose $t_m<10$ and $0<x(t_m)<1-t_m/10$ for some $m\geq 1$. From~\eqref{linear.CLsys}, we have 
\[
\dot{x}=-0.1-0.2x(t_m)<-0.1, \quad t\in [t_m,\min\{t_{m+1},r\}).
\] 
We derive that $x$ strictly decreases from $x(t_m)$ at $t=t_m$ to $x(t_{m+1})=x(t_m)/1.6$ at $t=t_{m+1}$ with $t_{m+1}<10$ and decreasing rate less than $-0.1$. Hence, $x(t_{m+1})<1-t_{m+1}/10$. Based on the above discussion, we conclude from mathematical induction that there are infinitely many event times on $[0,10]$, that is, control system~\eqref{linear.CLsys} exhibits Zeno behavior (see Fig.~\ref{fig3} for a numerical demonstration).

\begin{figure}[!t]\centering
\includegraphics[width=3.2in]{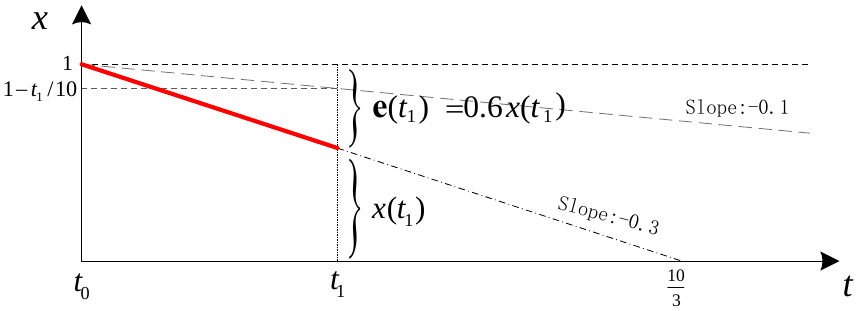}%
\caption{Trajectory of system~\eqref{linear.CLsys} on $[t_0,t_1]$. }
\label{Zenoproof}
\end{figure}

\begin{figure}[!t]\centering
\includegraphics[width=3.2in]{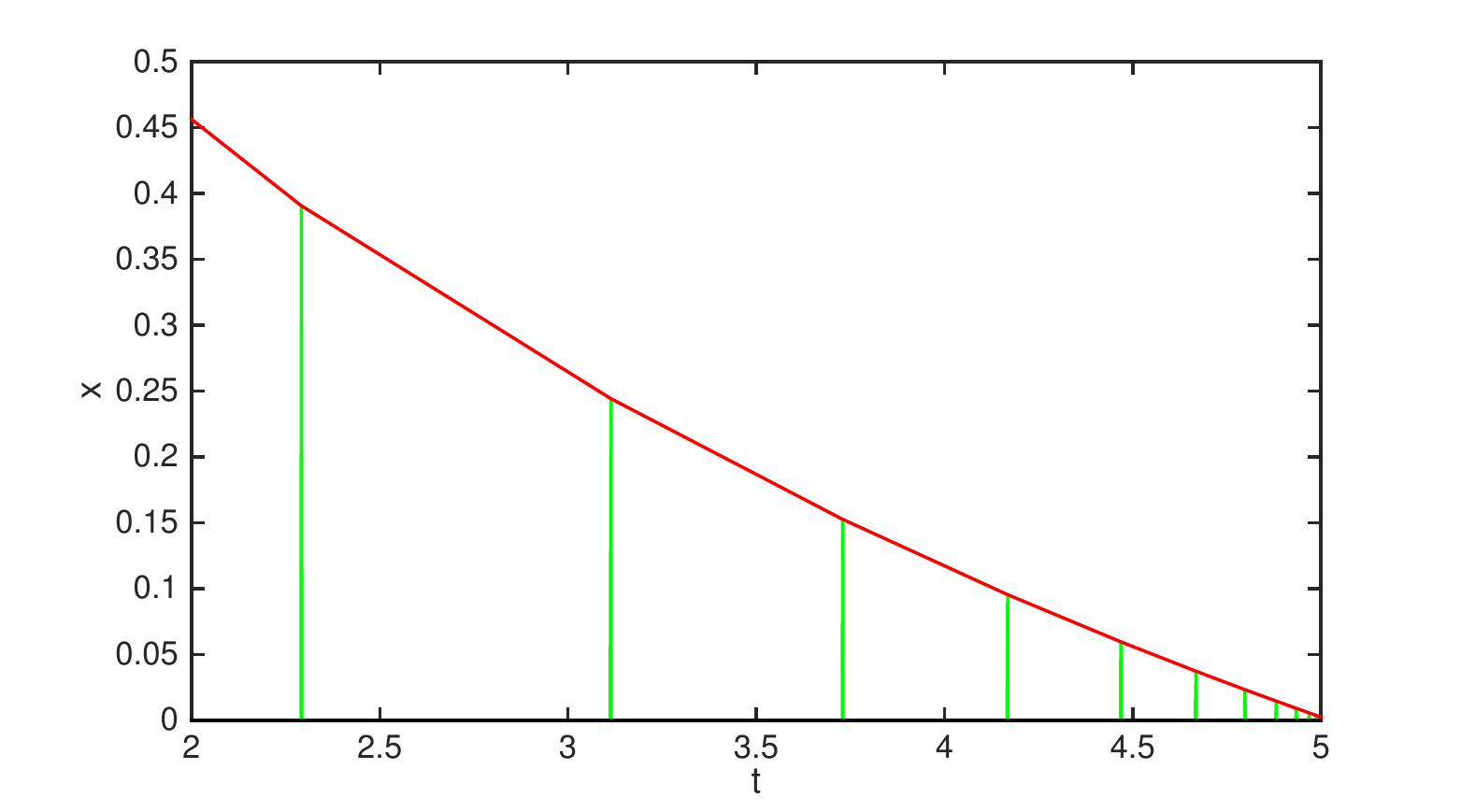}%
\caption{Trajectory of closed-loop system~\eqref{linear.CLsys} with the event times determined by~\eqref{et.time.linear1}. The green line segments indicate the event times. }
\label{fig3}
\end{figure}

For this particular system, its trajectory intersects with the time axis, and $\dot{x}$ is bounded. Therefore, {it takes less and less time for $|\mathbf{e}|$ to evolve from $0$ to $\sqrt{\sigma_0} |x|$ as $x$ getting closer and closer to $0$}. For linear scalar systems without time-delay, this property does not hold mainly because its trajectory does not go to zero in a finite time. Thus, the well-known result in \cite{PT:2007} for delay-free systems cannot be generalized to time-delay systems seamlessly in the sense of excluding Zeno behavior. The above example indicates that ruling out Zeno behavior cannot be guaranteed even for linear control systems with time-delay, though the sequence of event times determined by~\eqref{et.time} assures the GAS of the event-triggered control system~\eqref{linear.CLsys}.

Next, we will demonstrate the effectiveness of the proposed hybrid control algorithm.
\subsection{Application of the Hybrid Control Algorithm}\label{subsecIV-B}

Using the proposed hybrid strategy with~\textbf{\HEI}, the closed-loop system~\eqref{linear.sys} can then be written as a linear impulsive system:
\begin{eqnarray}\label{linear.isys}
\left\{\begin{array}{ll}
\dot{x}(t)=b x(t-r)+k x(t_{i}), \textrm{~for~} t\in[t_i,t_{i+1}) \cr
\Delta x(t_{i+1})=\beta x(t^-_{i+1}),~\textrm{~if~} t_{i+1}=t_i+h \cr 
x_{t_0}=\phi
\end{array}\right.
\end{eqnarray}
where the event-time sequence $\{t_i\}_{i\in\mathbb{N}}$ are chosen according to~\textbf{\HEI}, and the constants $\beta$ and $h$ are to be determined by using Theorem~\ref{Th.main}. It is not hard to observe that assumption~\ref{A1} holds with the given parameters and Lyapunov function $V(x)=x^2$. Next, we show that conditions~(i) and~(ii) of Theorem~\ref{Th.main} hold for system~\eqref{linear.isys}. When $t_{i+1}=t_i+h$, we can derive from the continuous dynamics of~\eqref{linear.isys} that, whenever $qV(x(t))\geq V(x(t+s))$ for all $s\in[-\tau,0]$, we have 
\begin{align}\label{precon3}
\dot{V}(x(t)) &\leq |b|[\varepsilon x^2+\varepsilon^{-1} x^2(t-r)] +|k|[\epsilon x^2+\epsilon^{-1} x^2(t_i)]\cr
           & =  (\varepsilon|b|+\epsilon|k|)x^2 + \varepsilon^{-1}|b|x^2(t-r) +\epsilon^{-1}|k|x^2(t_i)\cr
           & =  \bar{c}V(x(t)),
\end{align}
where $\bar{c}=2\sqrt{q}(|b|+|k|)$ and $\tau=\max\{r,h\}$. The first inequality of~\eqref{precon3} follows from Young's inequality with $\varepsilon=\sqrt{q/|b|}$ and $\epsilon=\sqrt{q/|k|}$. In the second inequality, we used the fact $x(t_i)=x(t-\delta(t))$ with $\delta(t)=t-t_i$ and $0\leq \delta(t)\leq h$ for $t\in[t_i,t_{i+1})$. Here, we have deemed $\delta(t)$ to be a time-varying delay when checking the Razumikhin-type condition (i) of Theorem~\ref{Th.main}. Using~\eqref{linear.isys}, we have that 
\[
V(x(t_{i+1}))=x^2(t_{i+1})=(1+\beta)^2 x^2(t^-_{i+1})=(1+\beta)^2 V(x(t^-_{i+1})),
\] 
that is, condition~(ii) of Theorem~\ref{Th.main} holds with $\rho=(1+\beta)^2$.

Therefore, we conclude from Theorem~\ref{Th.main} that if there exists a $q>1$ such that both~\eqref{inequ1} and condition (iii) of Theorem~\ref{Th.main} are satisfied, the closed-loop system~\eqref{linear.isys} is GAS, and $h$ is the lower bound of the inter-execution times.

To demonstrate the effectiveness of the proposed control algorithm and Theorem~\ref{Th.main}, let $q=3$, $h=0.666$ and $\beta=-0.293$ so that both~\eqref{inequ1} and condition (iii) of Theorem~\ref{Th.main} hold (See Remark~\ref{PSelection} for the procedure of parameter selections). Fig.~\ref{fig4} shows the stability of system~\eqref{linear.isys}. Actually, system~\eqref{linear.isys} is globally exponentially stable since $\alpha_1(|x|)=\alpha_2(|x|)=V(x)=x^2$ in condition (i) of Theorem~\ref{Th.ISS}. As discussed for system~\eqref{linear.CLsys}, the event-triggered control inputs are updated more and more frequently when the state $x$ gets closer and closer to zero. This explains why the impulses are generally activated around the intersections between the trajectory $x$ and the time axis in Fig.~\ref{fig4}. The reason for the existence of large inter-execution times is that it takes more time for $|\mathbf{e}|$ to evolve from zero at each event time to $\sqrt{\sigma_0} |x|$ at the next event time if $|x|$ is fairly large and/or $|\dot{x}|$ is relatively small.

\begin{remark}\label{PSelection}
Based on the Lyapunov function in Assumption~\ref{A1}, we have shown that conditions (i) and (ii) of Theorem~\ref{Th.main} hold with $\bar{c}=0.6\sqrt{q}$ and $\rho=(1+\beta)^2$. Parameters $h$, $q$, and $\rho$ should be selected so that the inequalities in condition (iii) are satisfied. First, we determine appropriate values of parameter $h$. It can be concluded from condition (iii) that $h<\frac{\ln (q)}{\bar{c}}=\frac{5\ln (q)}{3\sqrt{q}}$. Let $ G $ be defined as $G(q)=\frac{5\ln (q)}{3\sqrt{q}}$ for $q>1$. We then can see that the largest admissible $h$ is smaller than $\max_{q>1} G(q)$. 
It is easy to observe that $\max_{q>1} G(q)= G(e^2)\approx 1.2663$. Therefore, any $h<1.2663$ can be chosen as the lower bound of the inter-execution times; we have used $h=0.666$ in our simulation. Once parameter $h$ is obtained, we need to derive the values of $q$ so that $q>e^{\bar{c}h}$ with $h=0.666$ and $\bar{c}=0.6\sqrt{q}$, then parameter $\rho$ can be determined from condition~(iii). To identify suitable $q$, we define $H(q)=q-e^{\bar{c}h}$. We can see that equation $H(q)=0$ has two solutions $q_1\approx 1.6792$ and $q_2\approx 161.62$, and $H(q)>0$ for any $q_1<q<q_2$, that is, $q>e^{\bar{c}h}$ for all $q\in(q_1,q_2)$. Therefore, we can select $\rho$ such that $q_1<1/\rho<q_2$, that is, $0.0062<\rho<0.5955$, and then there exists $q\in (1/\rho,q_2)$ so that condition (iii) is satisfied. In the simulation, we chose $\rho=(1+\beta)^2\approx 0.4998<0.5955$. For different time-delay control systems, the dependence of $\bar{c}$ on parameter $q$ is different. However, the selections of parameter $h$ and $\rho$ can be conducted by following the selection procedure introduced above for linear control system~\eqref{linear.isys}.
\end{remark}

\begin{figure}[!t]
\centering
\includegraphics[width=3.2in]{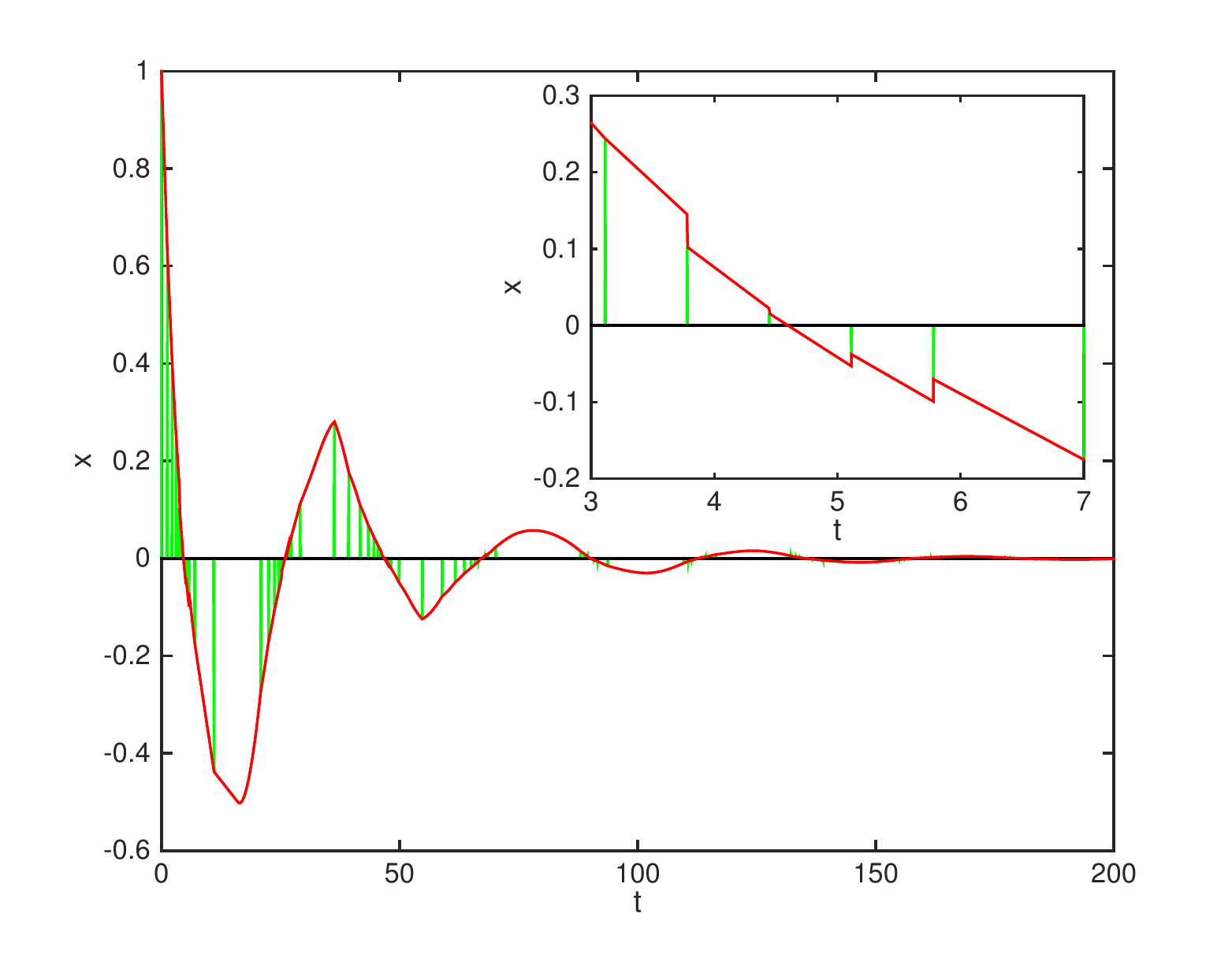}
\caption{Trajectory of system~\eqref{linear.isys}: the red curve represents the state $x$; the vertical green line segments correspond to the event times; the \textit{vertical} red line segments indicate the impulses in~\eqref{linear.isys}. A clear view of the impulses in system~\eqref{linear.isys} over the time interval $[3,7]$ is included within the figure.}
\label{fig4}
\end{figure}

\subsection{Comparison with Impulsive Control Strategy}

Without the feedback control input $u_1$, that is, $u_1\equiv 0$, we can derive from the previous discussion that conditions (i), (ii), and (iii) of Theorem~\ref{Th.main} hold with $\bar{c}$ replaced by $\bar{c}=2\sqrt{q}|b|$. According to the analysis in Remark~\ref{remark.impulse}, impulsive control system~\eqref{linear.isys} with $u_1\equiv 0$ is GAS under the impulse inputs considered in Fig.~\ref{fig4} (see Fig.~\ref{fig5} for a demonstration). It can be observed that far fewer impulse inputs are activated in the hybrid control system~\eqref{linear.isys} with \textbf{\HEI}. Moreover, less control updates are triggered in our hybrid control scheme since the inter-execution times are lower bounded by the inter-impulse time $h$, and this is verified by the observation that the event times in Fig.~\ref{fig4} is less than the impulse times in Fig.~\ref{fig5}.

\begin{figure}[!t]
\centering
\includegraphics[width=3.2in]{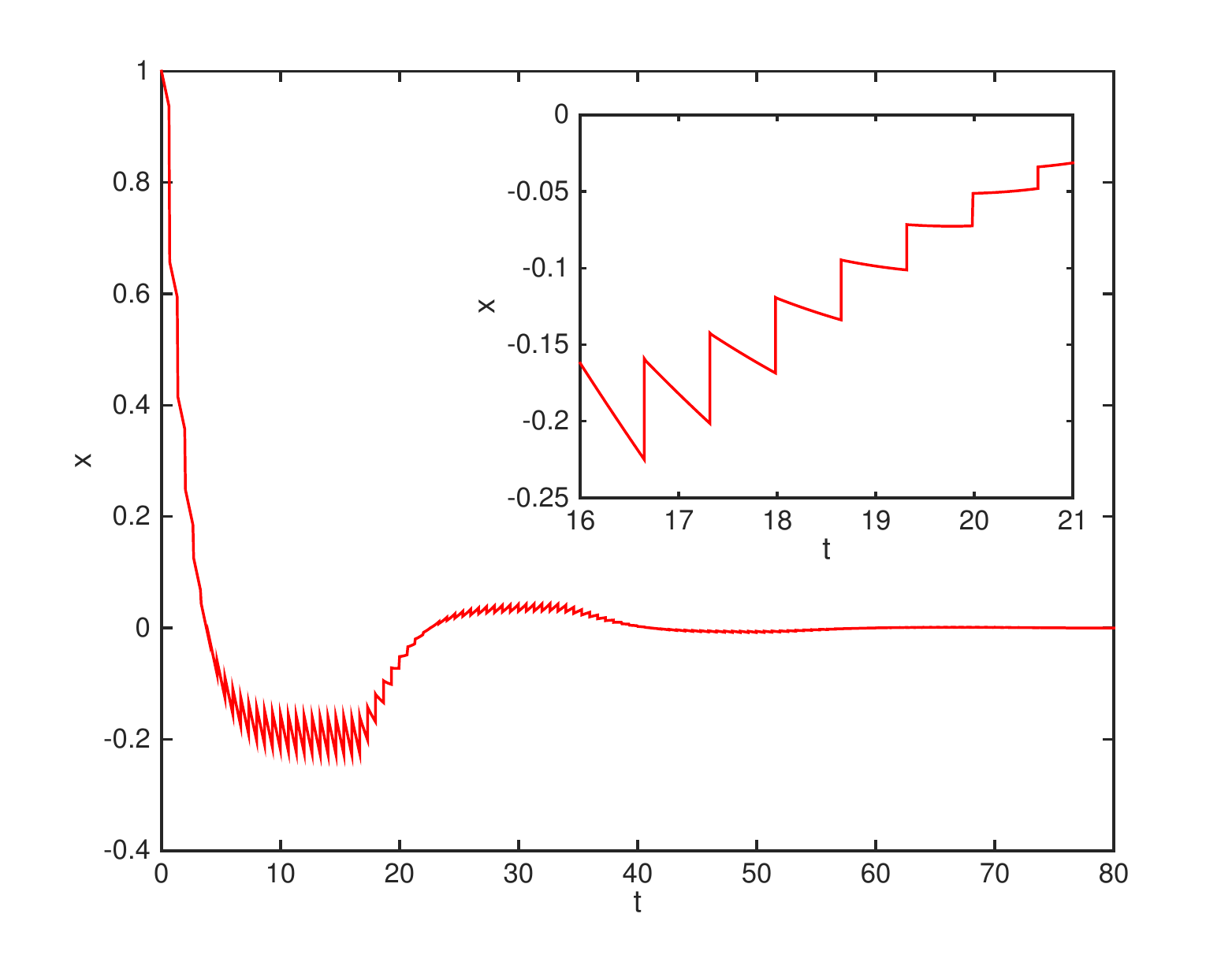}
\caption{Trajectory of impulsive control system~\eqref{linear.isys} with $k\equiv 0$ and $t_{i+1}-t_i=h=0.666$ ($i\in\mathbb{N}$). The inserted figure demonstrates a larger view of the trajectory over interval $[16,21]$. The \textit{vertical} line segments correspond to the state jumps.}
\label{fig5}
\end{figure}
\section{Conclusions}\label{conclusion}

We have studied hybrid stabilization problem of nonlinear time-delay systems. A hybrid event-triggered control algorithm has been proposed to stabilize the nonlinear systems with time-delay. To exclude Zeno behavior due to the presence of delay, our hybrid control algorithm has incorporated the impulsive control mechanism into the event-triggering scheme to guarantee the nonexistence of Zeno behavior. Future work includes applying our control algorithm to various related control problems, such as, consensus of multi-agent systems, distributed optimization, and seeking parallel algorithms based on the method of Lyapunov-Krasovskii functionals.

%\section*{References}

\end{document}